\documentclass[pra,aps,showpacs,twoside,twocolumn,superscriptaddress]{revtex4}
\usepackage{amsmath,amsfonts,amssymb,caption,hyperref,color,epsfig,graphics,graphicx,latexsym,mathrsfs,revsymb,theorem,url,verbatim,epstopdf}

\hypersetup{colorlinks,linkcolor={blue},citecolor={blue},urlcolor={red}} 

\usepackage{hyperref}

\theoremstyle{plain}
\newtheorem{definition}{Definition}
\newtheorem{proposition}[definition]{Proposition}
\newtheorem{lemma}{Lemma}

\newtheorem{theorem}{Theorem}
\newtheorem{corollary}{Corollary}
\newtheorem{conjecture}{Conjecture}

\newtheorem{remark}{Remark}
\newtheorem{example}{Example}
\newtheorem{question}[definition]{Question}

\def\squareforqed{\hbox{\rlap{$\sqcap$}$\sqcup$}}
\def\qed{\ifmmode\squareforqed\else{\unskip\nobreak\hfil
\penalty50\hskip1em\null\nobreak\hfil\squareforqed
\parfillskip=0pt\finalhyphendemerits=0\endgraf}\fi}
\def\endenv{\ifmmode\;\else{\unskip\nobreak\hfil
\penalty50\hskip1em\null\nobreak\hfil\;
\parfillskip=0pt\finalhyphendemerits=0\endgraf}\fi}
\newenvironment{proof}{\noindent \textbf{{Proof.~} }}{\qed}
\def\Dbar{\leavevmode\lower.6ex\hbox to 0pt
{\hskip-.23ex\accent"16\hss}D}
\makeatletter
\def\url@leostyle{%
  \@ifundefined{selectfont}{\def\UrlFont{\sf}}{\def\UrlFont{\small\ttfamily}}}
\makeatother
\def\bcj{\begin{conjecture}}
\def\ecj{\end{conjecture}}
\def\bcr{\begin{corollary}}
\def\ecr{\end{corollary}}
\def\bd{\begin{definition}}
\def\ed{\end{definition}}
\def\bea{\begin{eqnarray}}
\def\eea{\end{eqnarray}}
\def\bem{\begin{enumerate}}
\def\eem{\end{enumerate}}
\def\bex{\begin{example}}
\def\eex{\end{example}}
\def\bim{\begin{itemize}}
\def\eim{\end{itemize}}
\def\bl{\begin{lemma}}
\def\el{\end{lemma}}
\def\bma{\begin{bmatrix}}
\def\ema{\end{bmatrix}}
\def\bpf{\begin{proof}}
\def\epf{\end{proof}}
\def\bpp{\begin{proposition}}
\def\epp{\end{proposition}}
\def\bqu{\begin{question}}
\def\equ{\end{question}}
\def\br{\begin{remark}}
\def\er{\end{remark}}
\def\bt{\begin{theorem}}
\def\et{\end{theorem}}

\def\btb{\begin{tabular}}
\def\etb{\end{tabular}}

\newcommand{\nc}{\newcommand}

\def\r{\rho}

\def\ph{\varphi}

 \nc{\bbA}{\mathbb{A}} \nc{\bbB}{\mathbb{B}} \nc{\bbC}{\mathbb{C}}
 \nc{\bbD}{\mathbb{D}} \nc{\bbE}{\mathbb{E}} \nc{\bbF}{\mathbb{F}}
 \nc{\bbG}{\mathbb{G}} \nc{\bbH}{\mathbb{H}} \nc{\bbI}{\mathbb{I}}
 \nc{\bbJ}{\mathbb{J}} \nc{\bbK}{\mathbb{K}} \nc{\bbL}{\mathbb{L}}
 \nc{\bbM}{\mathbb{M}} \nc{\bbN}{\mathbb{N}} \nc{\bbO}{\mathbb{O}}
 \nc{\bbP}{\mathbb{P}} \nc{\bbQ}{\mathbb{Q}} \nc{\bbR}{\mathbb{R}}
 \nc{\bbS}{\mathbb{S}} \nc{\bbT}{\mathbb{T}} \nc{\bbU}{\mathbb{U}}
 \nc{\bbV}{\mathbb{V}} \nc{\bbW}{\mathbb{W}} \nc{\bbX}{\mathbb{X}}
 \nc{\bbZ}{\mathbb{Z}}

 \nc{\bA}{{\bf A}} \nc{\bB}{{\bf B}} \nc{\bC}{{\bf C}}
 \nc{\bD}{{\bf D}} \nc{\bE}{{\bf E}} \nc{\bF}{{\bf F}}
 \nc{\bG}{{\bf G}} \nc{\bH}{{\bf H}} \nc{\bI}{{\bf I}}
 \nc{\bJ}{{\bf J}} \nc{\bK}{{\bf K}} \nc{\bL}{{\bf L}}
 \nc{\bM}{{\bf M}} \nc{\bN}{{\bf N}} \nc{\bO}{{\bf O}}
 \nc{\bP}{{\bf P}} \nc{\bQ}{{\bf Q}} \nc{\bR}{{\bf R}}
 \nc{\bS}{{\bf S}} \nc{\bT}{{\bf T}} \nc{\bU}{{\bf U}}
 \nc{\bV}{{\bf V}} \nc{\bW}{{\bf W}} \nc{\bX}{{\bf X}}
 \nc{\bZ}{{\bf Z}}

\nc{\cA}{{\cal A}} \nc{\cB}{{\cal B}} \nc{\cC}{{\cal C}}
\nc{\cD}{{\cal D}} \nc{\cE}{{\cal E}} \nc{\cF}{{\cal F}}
\nc{\cG}{{\cal G}} \nc{\cH}{{\cal H}} \nc{\cI}{{\cal I}}
\nc{\cJ}{{\cal J}} \nc{\cK}{{\cal K}} \nc{\cL}{{\cal L}}
\nc{\cM}{{\cal M}} \nc{\cN}{{\cal N}} \nc{\cO}{{\cal O}}
\nc{\cP}{{\cal P}} \nc{\cQ}{{\cal Q}} \nc{\cR}{{\cal R}}
\nc{\cS}{{\cal S}} \nc{\cT}{{\cal T}} \nc{\cU}{{\cal U}}
\nc{\cV}{{\cal V}} \nc{\cW}{{\cal W}} \nc{\cX}{{\cal X}}
\nc{\cZ}{{\cal Z}}

\nc{\hA}{{\hat{A}}} \nc{\hB}{{\hat{B}}} \nc{\hC}{{\hat{C}}}
\nc{\hD}{{\hat{D}}} \nc{\hE}{{\hat{E}}} \nc{\hF}{{\hat{F}}}
\nc{\hG}{{\hat{G}}} \nc{\hH}{{\hat{H}}} \nc{\hI}{{\hat{I}}}
\nc{\hJ}{{\hat{J}}} \nc{\hK}{{\hat{K}}} \nc{\hL}{{\hat{L}}}
\nc{\hM}{{\hat{M}}} \nc{\hN}{{\hat{N}}} \nc{\hO}{{\hat{O}}}
\nc{\hP}{{\hat{P}}} \nc{\hR}{{\hat{R}}} \nc{\hS}{{\hat{S}}}
\nc{\hT}{{\hat{T}}} \nc{\hU}{{\hat{U}}} \nc{\hV}{{\hat{V}}}
\nc{\hW}{{\hat{W}}} \nc{\hX}{{\hat{X}}} \nc{\hZ}{{\hat{Z}}}

\nc{\hn}{{\hat{n}}}

\def\max{\mathop{\rm max}}
\def\min{\mathop{\rm min}}

\def\tr{\mathop{\rm Tr}}

\newcommand{\bra}[1]{\langle#1|}
\newcommand{\ket}[1]{|#1\rangle}
\newcommand{\proj}[1]{| #1\rangle\!\langle #1 |}

\def\Dbar{\leavevmode\lower.6ex\hbox to 0pt
{\hskip-.23ex\accent"16\hss}D}
\begin{document}
\title{Multi-linear Monogamy Relations for Multi-Qubit States}

\author{Xian Shi}\email[]
{shixian01@gmail.com}
\affiliation{School of Mathematics and Systems Science, Beihang University, Beijing 100191, China}

\author{Lin Chen}\email[]{linchen@buaa.edu.cn (corresponding author)}
\affiliation{School of Mathematics and Systems Science, Beihang University, Beijing 100191, China}
\affiliation{International Research Institute for Multidisciplinary Science, Beihang University, Beijing 100191, China}

%\author{Yi Shen}\email[]
%{yishen@buaa.edu.cn}
%\affiliation{School of Mathematics and Systems Science, Beihang University, Beijing 100191, China}
%
\author{Mengyao Hu}\email[]
{mengyaohu@buaa.edu.cn}
\affiliation{School of Mathematics and Systems Science, Beihang University, Beijing 100191, China}

%\author{Lijun Zhao}
%\affiliation{School of Mathematics and Systems Science, Beihang University, Beijing 100191, China}

%\author{Yumin Guo}
%\affiliation{School of Mathematical Sciences, Capital Normal University, Beijing 100048, China}

\date{\today}


\begin{abstract}
\indent The monogamy of entanglement means that entanglement cannot be freely shared. 
In 2014, Oliveira $et$ $al.$ [ Oliveira $et$ $al.$, Phys. Rev. A. \textbf{89}, 034303 (2014)] proposed a monogamy relation in the linear version and considered it in terms of entanglement of formation. Here we generalize the above version and consider a multi-linear monogamy relation for a multi-qubit system in terms of entanglement of formation and concurrence. Based on the above results, we present an entanglement criterion for genuine entangled states, also we consider the absolutely maximally entangled states and present what an absolutely maximally entangled state is for a three-qubit system. At last, we apply our results to a three-qubit pure state in terms of quantum discord.
	
\end{abstract}

%\Large

\maketitle 

%\tableofcontents

\section{Introduction}
\label{sec:int}
\indent Quantum entanglement is an essential feature of quantum mechanics. It plays an important role in quantum information and quantum computation theory \cite{horodecki2009quantum}, such as superdense coding \cite{bennett1992communication}, teleportation \cite{bennett1993teleporting} and the speedup of quantum algorithms \cite{shimoni2005entangled}. \\
\indent As a property of multipartite entanglement, monogamy of entanglement presents that entanglement cannot be shared arbitrarily among many parties, which is different from classical correlations \cite{terhal2004entanglement}.  This property has been applied on many areas in quantum information.  It can be applied to prove the security of quantum cryptography \cite{Pawlowski2010Security,Yang2016Measurement,tomamichel2013monogamy} and the bound of the regularization of its Holevo information for arbitrary channels \cite{gao2018heralded}. It can also be applied to distinguish inequivalent classes of pure states in a tripartite system \cite{coffman2000distributed,yu2009monogamy}. Recently, the authors showed there exists restrictions of indistinguishability for entangled systems due to monogamy relations \cite{karczewski2018monogamy}.\\
		\indent Mathematically, for a tripartite system with parties $A,B$ and $C,$ the general monogamy in terms of an entanglement measure $\mathcal{E}$ implies that the entanglement between $A$ and $BC$ satisfies
		\begin{align}\label{mr}
		\mathcal{E}_{A|BC}\ge \mathcal{E}_{AB}+\mathcal{E}_{AC},
		\end{align}
		here $\mathcal{E}_{AB}$ and $\mathcal{E}_{AC}$ means the entanglement between the parties $A,B$ and $A,C.$
		This relation was first proved for qubit systems in terms of the 2-tangle  \cite{coffman2000distributed,osborne2006general}.  Bai $et$ $al.$ showed that the inequality Eq. $(1)$ is valid in terms of the squared entanglement of formation (EoF) for $n$-qubit systems \cite{bai2014general}. Zhu $et$ $al.$  investigated the monogamy relations related to the concurrence and the entanglement of formation \cite{zhu2014entanglement}. Recently, the authors in \cite{zhu2015,san2016generalized} presented generalized monogamy relations, Jin $et$ $al.$ proposed tighter monogamy relations for $n$-qubit systems \cite{jin2018tighter}. Yu $et$ $al.$ utilized the conversion relation between the coherence and
		the entanglement to establish the monogamy inequalities
		for high-dimensional coherence-induced entanglement in
		terms of the relative entropy of entanglement and the negativity \cite{yu2019monogamy}. Zhang $et$ $al.$ studied the monogamy relations for  multi-qubit quantum systems in product norm \cite{zhang2020note}. \\
		\indent However, it is well known that the EoF ($E$) does not satisfy the inequality Eq. (\ref{mr}). In 2014, Oliveira $et$ $al.$ proposed a linear monogamy relation in terms of EoF  and numerically obtained the bound for a three-qubit system. This result indicates that entanglement cannot be freely shared in terms of EoF \cite{de2014monogamy}. In 2015, Liu $et$ $al.$ proved this bound analytically. There they also computed the bound of linear monogamy relation in terms of concurrence for a three-qubit system \cite{liu2015linear}. Moreover, Cornelio proposed another interesting monogamy relation in terms of the squared concurrence for three-qubit systems \cite{cornelio2013multipartite}. They called the relations multipartite monogamy relations.  \\
		\indent One of the motivations of this paper is to better understand the monogamy relations within the theory of multipartite entanglement. Although the authors in \cite{dur2000three} mentioned a similar function of a three-qubit pure state in terms of some entanglement measure,  there they aimed to investigate the robustness of a three-qubit pure state against loss of a qubit. Here we characterize the distribution of the entanglement for an $n$-qubit system in terms of EoF and concurrence.  In \cite{dur2000three}, the authors only showed the function numerically in terms of EoF and the bound of the function in terms of the squared concurrence among three-qubit pure states.  Crucially, we present multi-linear monogamy relation in terms of entanglement of formation for a three-qubit pure state analytically. We generalize this bound to a three-qubit mixed state in terms of EoF and concurrence. Also we present only the LU-equivalent class of W state can reach the upper bound among three-qubit mixed states. That is, this can be seen to detect whether a three-qubit pure state is W state. Due to the importance of the W state in quantum computation and communication \cite{joo2003quantum,ng2014quantum,yu2014obtaining,vijayan2020robust}, this result is meaningful.\\
		\indent In this work, we consider a multi-linear monogamy relation in terms of EoF and concurrence for a multi-qubit system. We present that the W state is the unique state that can reach the upper bound of multi-linear monogamy relations in terms of concurrene and EoF up to the local unitary transformations (LU). We also present the condition when the states reach the minimum of the multi-linear monogamy relation in terms of concurrence. At last, we present some applications of our results to build an entanglement criterion and consider the absolutely maximally entangled states for a three-qubit system mainly. We also get a similar bound for the discord of three-qubit pure states.\\
		\indent This article is organized as follows. First we review the preliminary knowledge needed. Then we prove our main results. We present multi-linear monogamy relations in terms of EoF and concurrence.  We also present some applications of our results on the entanglement witness. At last, based on the relation between the EoF and the discord, we present a similar result for the sum of all bipartite quantum discord for a three-qubit pure state. 
\section{Preliminaries}
\label{sec:pre}
\indent An $n$-partite pure state $\ket{\psi}_{A_1A_2\cdots A_n}$ is full product if it can be written as
\begin{align}
\ket{\psi}_{A_1A_2\cdots A_n}=\ket{\phi_1}_{A_1}\ket{\phi_2}_{A_2}\cdots\ket{\phi_n}_{A_n},
\end{align}
otherwise, it is entangled. A multipartite pure state is called genuinely entangled if 
\begin{align}
\ket{\psi}_{A_1A_2\cdots A_n}\ne \ket{\phi}_{S}\ket{\varphi}_{\overline{S}},
\end{align}
for any bipartition $S|\overline{S},$ here $S$ is a subset of $\boldsymbol{A}=\{A_1,A_2,\cdots,A_n\}$, $\overline{S}=\boldsymbol{A}-S.$\\
	\indent Assume $\ket{\psi}_{AB}$ is a bipartite pure state. Due to the Schmidt decomposition, $\ket{\psi}_{AB}$ can always be written as $$\ket{\psi}_{AB}=\sum_i \sqrt{\lambda_i}\ket{i}_A\ket{i}_B,$$ here $\lambda_i\ge 0,\sum_i\lambda_i=1,$ $\{\ket{i}_{A(B)}\}$ is an orthonormal basis of the Hilbert space $A(B).$  First we recall the EoF. The EoF of $\ket{\psi}_{AB}$ is given by
		\begin{align}\label{Ep}
		E(\ket{\psi}_{AB})=S(\rho_A)=-\sum \lambda_i\log\lambda_i,
		\end{align} 
		here $\lambda_i$ are the eigenvalues of $\rho_A=\tr_B\ket{\psi}_{AB}\bra{\psi}.$ For a mixed state $\rho_{AB},$ the EoF is defined by the convex roof extension method,
		\begin{align}\label{Em}
		E(\rho_{AB})=\min_{\{p_i,\ket{\phi_i}_{AB}\}}\sum_i p_i E(\ket{\phi_i}_{AB}),
		\end{align} 
		where the minimum is taken over all the decompositions of $\rho_{AB}=\sum_i p_i\ket{\phi_i}_{AB}\bra{\phi_i}$ with $p_i\ge 0$ and $\sum p_i=1.$\\
		\indent The other important entanglement measure is the concurrence ($C$). The concurrence of a pure state $\ket{\psi}_{AB}$ is defined as 
		\begin{align}\label{Cp}
		C(\ket{\psi}_{AB})=\sqrt{2(1-\tr\rho_A^2)}=\sqrt{2(1-\sum_i\lambda_i^2)}.
		\end{align}
		For a mixed state $\rho_{AB},$ it is defined as
		\begin{align}\label{Cm}
			C(\rho_{AB})=\min_{\{p_i,\ket{\phi_i}_{AB}\}}\sum_i p_i C(\ket{\phi_i}_{AB}),
		\end{align}
			where the minimum takes over all the decompositions of $\rho_{AB}=\sum_i p_i\ket{\phi_i}_{AB}\bra{\phi_i}$ with $p_i\ge 0$ and $\sum p_i=1.$\\
			\indent For a two-qubit mixed state $\rho_{AB},$ Wootters derived an analytical formula \cite{wootters1998entanglement}:
			\begin{align}\label{E2}
				E(\rho_{AB})=&h(\frac{1+\sqrt{1-C_{AB}^2}}{2}), \\
				h(x)=&-x\log_2x-(1-x)\log_2(1-x),\label{h}\\
			C_{AB}=&\max\{\sqrt{\mu_1}-\sqrt{\mu_2}-\sqrt{\mu_3}-\sqrt{\mu_4},0\},\label{C2}
			\end{align}
		 here the $\mu_1,\mu_2,\mu_3,\mu_4,$ are the eigenvalues of the matrix $\rho_{AB}(\sigma_y\otimes\sigma_y)\rho_{AB}^{*}(\sigma_y\otimes\sigma_y)$ with nonincreasing order.\\
\section{Main Results}
\label{sec:resultsummary}
\indent For a three-qubit pure state $\ket{\psi}_{ABC}$, the pairwise correlations are described by the reduced density operators $\rho_{AB},\rho_{BC}$ and $\rho_{CA}.$ In 2014, Oliveira $et$ $al.$ \cite{de2014monogamy}  numerically presented the following inequality is valid for a $3$-qubit pure state in terms of EoF and concurrence, 
\begin{align}
E_{A|B}+E_{A|C}\le \lambda, \label{lmoe}
\end{align}
here $\lambda$ is a constant, when $E$ is EoF, they conjectured $\lambda=1.2018$. In 2015, the authors in \cite{liu2015linear} proved the above inequality for a $3$-qubit pure state in terms of EoF analytically, there they denoted the above inequality as the linear monogamy relation.\\
\indent From the (\ref{lmoe}), we find that although the EoF doesnot satisfy  (\ref{mr}) for 3-qubit generic states, the entanglement cannot be freely shared in terms of EoF. Here we mainly consider a new linear monogamy relation which we call it multi-linear monogamy relation. The main difference between ours and the linear monogamy relations is that the left hand side takes over all bipartitions within the multipartite entanglement. For the $3$-qubit states, it means in terms of some entanglement measure $\mathcal{E}$, the following inequality is valid,
\begin{align}
M\mathcal{E}=\mathcal{E}_{A|B}+\mathcal{E}_{A|C}+\mathcal{E}_{B|C}\le\nu.
\end{align}
Here $\nu$ is a constant. We can also generalize the relations to $n$-qubit states $\rho_{A_1A_2\cdots A_n},$ we denote the following inequality in terms of some entanglement measure $\mathcal{E}$ as the multi-linear monogamy relation
\begin{align}
\sum_{i< j}\mathcal{E}_{i|j}\le \eta.
\end{align}
Here $\eta$ is a constant.
\subsection{Multipartite linear monogamy relations in terms of EoF}	
\indent In this subsection, we first present a theorem on the multi-linear monogamy relation in terms of EoF  for a three-qubit pure state.  	
\begin{theorem}\label{meofw}
	For a three-qubit pure state, the W state reaches the upper bound $c_{max}=3h(\frac{1}{2}+\frac{\sqrt{5}}{6})$ of multi-linear monogamy relation in terms of EoF.
	\end{theorem}	
	The proof of Theorem \ref{meofw} is in the APPENDIX \ref{A}.\\
\indent We can extend this result to the mixed state $\rho_{ABC}$. Assume that $\{s_h,\ket{\phi_h}_{ABC}\}$ is a decomposition of $\rho_{ABC},$ then we have
		\begin{align}
		\hspace{2mm} &E(\rho_{AB})+E(\rho_{AC})+E(\rho_{BC})\nonumber\\
		=&\sum_i p_iE(\ket{\phi_i}_{AB})+\sum_j q_jE(\ket{\theta_j}_{AC})+\sum_k r_k E(\ket{\zeta_k}_{BC})\nonumber\\
		\le& \sum_h s_h(E(\rho_{AB}^h)+E(\rho_{AC}^h)+E(\rho_{BC}^h))\nonumber\\
		\le &\sum_h s_h\times c_{max}=c_{max}.\label{mm}
		\end{align}
		Here we assume that in the first equality, $\{p_i,\ket{\phi_i}\},$ $\{q_j,\ket{\theta_j}\}$ and $\{r_k,\ket{\zeta_k}\}$ are the optimal decompositions of $\rho_{AB},$ $\rho_{AC}$ and $\rho_{BC}$ in terms of the EoF correspondingly. The first equality is due to the definition of the EoF for the mixed states, the second inequality is due to the equality (\ref{mp}). In the first inequality,  we denote  $\tr_C\ket{\phi^{h}}\bra{\phi^h}=\rho_{AB}^h,\tr_B\ket{\phi^{h}}\bra{\phi^h}=\rho_{AC}^h,\tr_A\ket{\phi^{h}}\bra{\phi^h}=\rho_{BC}^h,$ \\
	
				\indent   For a three-qubit pure state, D$\ddot{u}$r $et$ $al.$ \cite{dur2000three} showed that there are two inequivalent kinds of genuinely entangled states, $i.e.$ the W-class states and the Greenberger-Horne-Zeilinger (GHZ)-class states.  The W-class states $\ket{\psi}$ are all LU equivalent to the following states:
		\begin{align}
           \ket{\phi}=r_0\ket{000}+r_1\ket{001}+r_2\ket{010}+r_3\ket{100},\label{LUW}
		\end{align}
		where $r_1,r_2,r_3>0,$ and $\sum_{i=0}^3|r_i|^2=1.$ From simple computation, we have $C^2(\rho_{AB})=4|r_2r_3|^2,C^2(\rho_{AC})=4|r_1r_3|^2,C^2(\rho_{BC})=4|r_1r_2|^2.$ We see that the function $E(\rho_{AB})+E(\rho_{AC})+E(\rho_{BC})$ ranges over $(0,c_{max}]$ for the W class states. When $\ket{\psi}=\frac{\ket{000}+\ket{111}}{\sqrt{2}},$ $E(\rho_{AB})+E(\rho_{AC})+E(\rho_{BC})=0,$ and as the GHZ class states is dense \cite{acin2001classification}, the function $E(\rho_{AB})+E(\rho_{AC})+E(\rho_{BC})$ ranges over $[0,c_{max})$. 

		\subsection{Multipartite linear monogamy relations in terms of concurrence}
		\indent In this subsection, we present a theorem on the multi-linear monogamy relation in terms of the concurrence  for a three-qubit pure state $\ket{\psi}_{ABC}$.  
		 \begin{lemma}\label{mulc}
			Up to the local unitary transformations, the W state is the unique state that can reach the upper bound in terms of the function $MC(\psi)=C_{AB}+C_{BC}+C_{AC}$ for a three-qubit pure state.\\
		\end{lemma}
	\indent We place the proof of the lemma \ref{mulc} in the APPENDIX \ref{B}

\indent By the similar method we present under the Theorem \ref{meofw}, we can also extend the above results on the mixed states.\\ \indent Next we present an example on the multi-linear monogamy relation in terms of concurrence for a three-qubit mixed state.
\begin{example}
	$$\rho=p_1\ket{W}\bra{W}+p_2\ket{\overline{W}}\bra{\overline{W}}.$$ Here we denote that $\ket{\overline{W}}=\frac{1}{3}(\ket{110}+\ket{101}+\ket{011}).$
	\end{example}

	Through simple computation,
	\begin{align*}
	\rho_{AB}=&\rho_{AC}=\rho_{BC}\\
\rho_{AB}=&	\frac{p_1}{3}\ket{00\bra{00}}+\frac{1}{3}(\ket{01}+\ket{10})(\bra{10}+\bra{01})+\frac{p_2}{3}\ket{11}\bra{11},
	\end{align*} we have  $$C(\rho_{AB})=C(\rho_{AC})=C(\rho_{BC})=\frac{2-2\sqrt{p_1p_2}}{3},$$ if $p_i>0,i=1,2,$ we have $MC(\rho_{ABC})<2.$\\

\indent The Lemma \ref{mulc} can be generalized to the three-qubit mixed states.
\begin{theorem}\label{mulcm}
	Up to the local unitary transformations, the W state is the unique state that can reach the upper bound in terms of the function $MC(\cdot)$ for a three-qubit mixed state.\\
\end{theorem}
\indent The proof of Theorem \ref{mulcm} is placed in the APPENDIX \ref{C}.\\
\indent Next we present a necessary and sufficient condition when the function $MC(\cdot)$ attains the minimum 0.
\begin{theorem}\label{6}
	Assume $\ket{\psi}$ is a three-qubit pure state, then $MC(\ket{\psi})=0$ if and only if $\ket{\psi}$ can be represented as $\ket{\psi}=r_0\ket{000}+r_1\ket{111}$ up to local unitary operations when $0\le r_0,r_1\le 1.$
 \end{theorem}
\indent The proof of Theorem \ref{6} is placed in the APPENDIX \ref{D}.\\
\begin{theorem}\label{t5}
	Up to the local unitary transformations, the W state is the unique state that can reach the upper bound in terms of the function $E(\rho_{AB})+E(\rho_{AC})+E(\rho_{BC})$ for a three-qubit mixed state.
\end{theorem}
\par Theorem \ref{t5} can be proved in a similar process with the proof of theorem \ref{mulcm}.\\
\indent Next we consider the multi-linear monogamy relations for the $n$-qubit W-class states. These states were first proposed by \cite{Kim2008Entanglement} in order to study the monogamy relations in terms of convex roof extended negativity for higher dimensional systems.
\begin{example}
	$$\ket{\phi}_{A_1A_2\cdots A_n}=\sqrt{p}\ket{GW}_n+\sqrt{1-p}\ket{0}_n.$$
	Here we assume $\ket{GW}=a_1\ket{10\cdots 0}+a_2\ket{010\cdots}+\cdots+a_n\ket{00\cdots 1},$ $\sum_i |a_i|^2=1.$
\end{example}\label{7}
\indent Through simple computation, we have $C(\rho_{A_1A_i})=2p|a_1a_i|$. \begin{align}\label{scgw}
MC(\ket{\phi})=&2p\sum_{i<j}|a_ia_j|\nonumber\\
=&p[(\sum_{i}|a_i|)^2-\sum_i|a_i|^2]\nonumber\\
=&p[(\sum_{i}|a_i|)^2-1].
\end{align}
By the method of Lagrange multiplier, we see when $a_i=\frac{1}{\sqrt{n}},p=1$, that is, when $\ket{\phi}=\ket{W},$ the value in (\ref{scgw}) attains the maximum.\\
\indent In \cite{koashi2000entangled}, the authors presented that for an $n$-qubit symmetric pure state $\ket{\phi}$, the maximal value between any pair of qubits in terms of concurrence is $\frac{2}{n},$ and when $\ket{\phi}=\ket{W}$, it attains the maximum.   Then we may propose a conjecture.

\begin{conjecture}
\label{cj:nqubitMES=w}
	 For an $n$-qubit genuinely entangled pure state $\ket{\phi}$, the maximum $MC(\phi)$ is attained when $\ket{\phi}=\ket{W}.$
\end{conjecture}\label{c1}
\begin{remark}Under the Conjecture \ref{cj:nqubitMES=w}, we can generalize the above results to $n$-qubit mixed states.\\
\indent First, we prove that when $\ket{\phi}_{A_1A_2\cdots A_n}$ is an $n$-qubit pure state , the maximum of $MC(\phi)$ is attained when $\ket{\phi}=\ket{W},$ that is, $\max_{\phi}MC(\ket{\phi})=n-1.$

	If $\ket{\phi}$ is not genuinely entangled,  we can always assume that $\ket{\phi}_{A_1A_2\cdots A_n}$ is biseparable, $i.e.$ $\ket{\phi}_{A_1A_2\cdots A_n}=\ket{\theta_1}_{A_1A_2\cdots A_m}\ket{\theta_2}_{A_{m+1}A_{m+2}\cdots A_n},$ here $\ket{\theta_i},$ $i=1,2$ are genuinely entangled. As $\ket{\phi}$ is biseparable, 
	\begin{align}
	&MC(\ket{\phi})\nonumber\\
	=&MC(\ket{\theta_1})+MC(\ket{\theta_2})\nonumber\\
	\le& m-1+n-m-1\nonumber\\
	=&n-2\nonumber\\
	<&n-1.
	\end{align}

\indent Then by a similar proof of Theorem \ref{mulcm} and the statement above, we can get the results on mixed states: when $\rho_{A_1A_2\cdots A_n}$ is an $n$-qubit mixed state, up to the local unitary transformations,   the W state is the unique state that can reach the upper bound in terms of the function $\sum_{i<j}C_{ij}$.
\end{remark}

\indent Then we pick $10^5$ four qubit and five qubit pure states randomly and compute their $MC(\cdot)$,  these results may verify the Conjecture \ref{cj:nqubitMES=w} numerically.

\begin{figure}
	\centering
	\includegraphics[width=100mm]{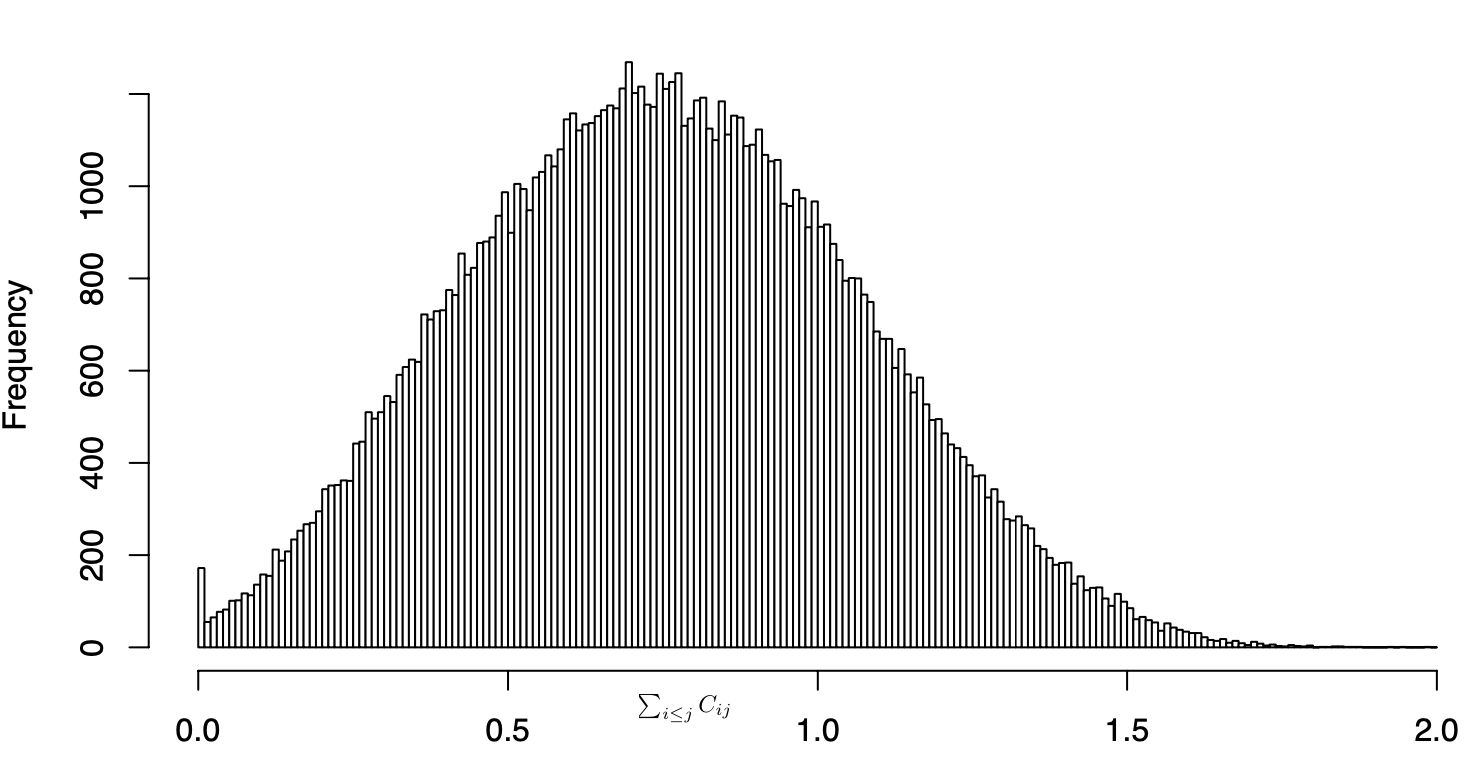}\\
	\caption{In this figure,  we present the frequency of the function $\sum_{i\le j}C_{ij}$ for random pure states of four qubit states.}\label{fig1}
\end{figure}
In Fig. \ref{fig1}, we present a histogram of the value of $\sum_{i\le j}C_{ij}$ for random pure states of four qubits sampled uniformly. Here we find the function $MC(\cdot)$ mainly distributes in the section $[0,1.8],$ and in the section [1.8,2], there are few states. Fig. \ref{fig1} supports Conjecture \ref{cj:nqubitMES=w} for $n=4$.\\
\begin{figure}
	\centering
	\includegraphics[width=90mm]{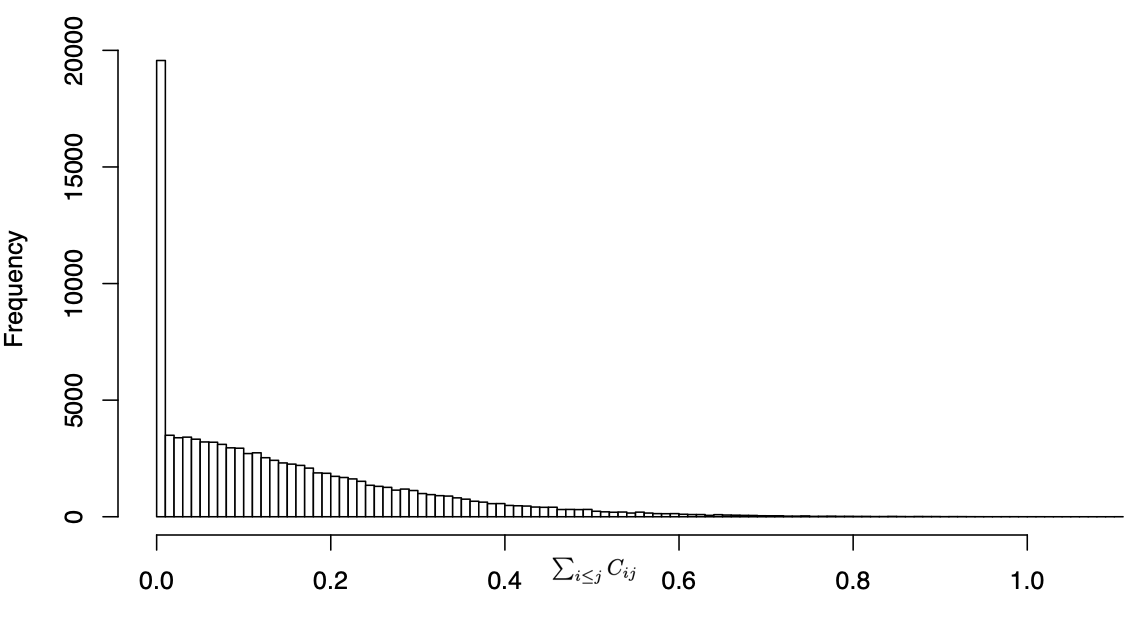}\\
	\caption{In this figure,  we present the frequency of the function $\sum_{i\le j}C_{ij}$ for random pure states of five qubit states.}\label{fig2}
\end{figure}
\indent In Fig. \ref{fig2}, we present a histogram of the value of $\sum_{i\le j}C_{ij}$ for random pure states of five qubits sampled uniformly. From the figure, we have that the sum of $MC(\cdot)$ mainly distributes in the section $[0,0.6]$.  In \cite{bai2007multipartite}, the authors considered the multipartite correlations in four-qubit pure states. Here through the Fig \ref{fig2}, we have that the quantity of the bipartite correlations of most five-qubit pure states is few, then it seems that comparing with the separable states, the set of the entangled states for five-qubit pure states are bigger.
\par
At the last of this section, we consider the $MC(\cdot)$ for a class of pure states in a system with more qubits  studied in [34]. They are useful kinds of entanglement states for quantum teleportation and error correction,
\begin{align*}
\ket{\psi}=a\ket{GHZ}_m\ket{W}_n+b\ket{W}_m\ket{GHZ}_n.
\end{align*}
Here $|a|^2+|b|^2=1,$ and $m,n\ge 2$. Due to the shape of $\ket{\psi}$, we have the set of the bipartite reduced density matrices for the pure state $\ket{\psi}$ consist of three kinds:
\begin{widetext}
 \begin{align}
 \rho^1=&\frac{|a|^2}{2}(\ket{00}\bra{00}+\ket{11}\bra{11})+\frac{|b|^2}{m}[(m-2)\ket{00}\bra{00}+(\ket{01}+\ket{10})(\bra{01}+\bra{10})],\nonumber\\
 \rho^2=&\frac{|b|^2}{2}(\ket{00}\bra{00}+\ket{11}\bra{11})
 +\frac{|a|^2}{n}[(n-2)\ket{00}\bra{00}+(\ket{01}+\ket{10})(\bra{01}+\bra{10})],\nonumber\\
 \rho^3=&(\frac{a}{\sqrt{2n}}\ket{01}+\frac{b}{\sqrt{2m}}\ket{10})(\frac{\overline{a}}{\sqrt{2n}}\bra{01}+\frac{\overline{b}}{\sqrt{2m}}\bra{10})+[\frac{(n-1)|a|^2}{2n}+\frac{(m-1)|b|^2}{2m}]\ket{00}\bra{00}+\frac{(n-1)|a|^2}{2n}\ket{10}\bra{10}\nonumber\\
 +&\frac{|a|^2}{2n}\ket{11}\bra{11}+\frac{(m-1)|b|^2}{2m}\ket{01}\bra{01}+\frac{|b|^2}{2m}\ket{11}\bra{11},
 \end{align}
 \end{widetext}
 then we have 
 \begin{align}
 &C(\rho^1)\nonumber\\
 =&\max\{0,\frac{2|b|^2}{m}-2\sqrt{\frac{|a|^2}{2}(\frac{|a|^2}{2}+\frac{m-2}{m}|b|^2)}\}\nonumber\\=&\begin{cases}
 0& |a|^2\in(g(m),1],\label{w1}\\
 \frac{2|b|^2}{m}-2\sqrt{\frac{|a|^2}{2}(\frac{|a|^2}{2}+\frac{m-2}{m}|b|^2)}& |a|^2\in[0,g(m)] 
 \end{cases}\\
 &C(\rho^2)\nonumber\\=
 &\max\{0,\frac{2|a|^2}{n}-2\sqrt{\frac{|b|^2}{2}(\frac{|b|^2}{2}+\frac{n-2}{n}|a|^2)}\}\nonumber\\=&\begin{cases}
 0& |a|^2\in[0,h(n)]\\
 \frac{2|a|^2}{n}-2\sqrt{\frac{|b|^2}{2}(\frac{|b|^2}{2}+\frac{n-2}{n}|a|^2)}& |a|^2\in[h(n),1]
 \end{cases}\label{w2}\\
 &C(\rho^3)=0.
 \end{align}
 Here $g(m)=\frac{m^2-2m+4-m\sqrt{m^2-4m+8}}{(m-2)^2},h(n)=\frac{n\sqrt{n^2-4n+8}-2n}{(n-2)^2}$.\par
 Let $x=|a|^2$, $f(x)=MC(\ket{\psi})$.
First we will compute the maximum when  $(m,n)\in\{(2,2),(2,3),(3,2),(3,3),(2,4),(4,2)\}$.\par
When $m=n=2,$ we have 
\begin{align}
f(x)=\begin{cases}
1-2x&x\in[0,1/2]\\
2x-1& x\in[1/2,1]\\
\end{cases},
\end{align}
then the maximum is 1, when $x=0$ or $1.$\par
When $m=2,n=3$, we have 
\begin{align}
f(x)=\begin{cases}
1-2x& x\in(0,1/2)\\
0&x\in [1/2,3\sqrt{5}-6]\\
{2x}-\sqrt{3x^2-12x+9}& x\in(3\sqrt{5}-6,1],
\end{cases}
\end{align}
then we have when $x=1$, $f(x)$ gets the maximum $2$.\par
When $m=3, n=2$, we have
\begin{align}
f(x)=\begin{cases}
2(1-x)-\sqrt{3x(x+2)}&x\in [0,7-3\sqrt{5}]\\
0& x\in [7-3\sqrt{5},\frac{1}{2}]\\
2x-1& x\in [1/2,1],
\end{cases}
\end{align}
when $x=0$, $f(x)$ gets the maximum 2.\par
When $m=n=3,$ we have
\begin{align}
f(x)=\begin{cases}
3(1-x)-\sqrt{3{x(x+2)}}&x\in [0,7-3\sqrt{5}]\\
0&x\in (7-3\sqrt{5}, 3\sqrt{5}-6)\\
2x-\sqrt{3x^2-12x+9}& x\in(3\sqrt{5}-6,1],
\end{cases}
\end{align}
when $x=0$ or $x=1,$ $f(x)$ gets the maximum, $2.$\\
\indent When $m=4,n=2$, we have
\begin{align}
f(x)=\begin{cases}
3-3x-6\sqrt{x}& x\in[0,3-2\sqrt{2}]\\
0& x\in [3-2\sqrt{2},1/2]\\
2x-1& x\in [1/2,1],
\end{cases}
\end{align}
when $x=0$, $f(x)$ gets the maximum, $3.$\\
\indent When $m=2,$ $n=4$, we have
\begin{align}
f(x)=\begin{cases}
1-2x& x\in [0,1.2]\\
0&x\in(1/2,2\sqrt{2}-2)\\
3x-12\sqrt{\frac{1}{4}-\frac{x}{4}}&x\in [2\sqrt{2}-2,1]
\end{cases},
\end{align}
then $\max_{x\in[0,1]}f(x)=3.$\\
\indent Next
from simple computation, we have $\forall m,n\ge 3,$ $h(n)\ge g(m),$
then 
\begin{align}
&f(x)\nonumber\\
=&\begin{cases}
(m-1)(1-x)-2\sqrt{\frac{4-m}{4m}x^2+\frac{m-2}{2m}x}&x\in [0,g(m)]\\
0& x\in[g(m),h(n)]\\
(n-1)x-n(n-1)\sqrt{\frac{1}{4}-\frac{x}{n}+\frac{4-n}{4n}x^2}& x\in (g(n),1)
\end{cases}
\end{align}
when $x\in[0,g(m)],m\ge 4,$ $f(x)$ is monotone decreasing, then when $x=0,$ $f(x)=m-1,$ when $x\in (g(n),1)$ $f(x)$ is monotone increasing,  then $\max_{x\in(g(n,1))}f(x)=n-1,$ that is, when $m\ge 4,$ $f(x)=\max(m-1,n-1)$. \par
In the next section, we present some applications of our results.
\section{application}
\indent The structure of a multipartite entanglement system is complex. In this section, we apply our results above on the genuine entanglement detection. We also make some comments on the abosolutely maximally entangled states (AMES), and we apply Theorem \ref{6} to present when a pure state is AMES in a three-qubit system. 

\subsection{An entanglement criterion for genuine entangled states}
\indent On the other hand, an important problem in entanglement theory is to determine whether a multipartite state is genuinely entangled, biseparable or fully separable. An widely accepted method for attacking the problem is to construct entanglement witnesses (EWs)  \cite{horodecki2009quantum}.  The EW $W$ is an Hermitian operator when $\tr(W\sigma)\ge 0$ for every biseparable state $\sigma$, and $\tr(W\rho)<0$ for some entangled state $\rho$. The EW is a theoretical and experimental method compared with mathematical criteria, such as positive partial transpose (PPT) \cite{peres1996separability} and computable cross norm \cite{rudolph2005further}. For a review we refer readers to the papers  \cite{horodecki2009quantum, guhne2009entanglement}. EWs have been constructed to detect the entanglement of many physically realizable states, such as the GHZ diagonal states \cite{chen2018precise,chen2019hierarchy}, GHZ-like state \cite{zhao2019efficient}, noisy Dicke states \cite{chen2020noise}. In the following we connect EWs to the $MC(\phi)$. 

In practice, we need analyze the change of $MC(\phi)$ of $n$-qubit pure states $\ket{\ph}$ under the white noise. Let  $\r(p)=p\proj{\ph}+(1-p){1\over2^n}I_{2^n}$ with $p\in(0,1)$. By definition and simple computation, one can show that the $MC(\r)$  monotonically increases with $p$. As an example, we assume that $\ket{\ph}$ is the $n$-qubit W state. Hence
	\begin{align}
	\r(p)=p\proj{W_n}+(1-p)\frac{I}{2^n},\label{rp}
	\end{align}
	here we can find that the state $(\ref{rp})$ is symmetric. From computation, we have 
	\begin{widetext}
	\begin{align}
	C(\rho_{i\ne j})=\max(0,\frac{2p}{n}-2\sqrt{\frac{(1-p)(n+3np-8p)}{16n}}),
	\end{align}
	\end{widetext}
	that is, when 
	\begin{align}
	16p^2-n^2(1-p)^2+4np(n-2)(p-1)>0,\label{wge0}
	\end{align}
	\begin{widetext} 
		$$MC(\rho)=(n-1)p-(n-1)\sqrt{n(1-p)(n+3np-8p)}.$$
		\end{widetext} When $n=3$, from the analysis above, we have when $p\ge \frac{\sqrt{155}-5}{8},$ the state $\r(p)$ is a genuinely entangled state. In \cite{chen2020noise}, the authors showed when $p\le \frac{\sqrt{3}}{8+\sqrt{3}},$ the state $\r(p)$ is fully separable, there the authors present an optimal entanglement witness $\widehat{W}$ for the $\r(p)$ when $n=3,$ the witness $\widehat{W}$ is written as
	\begin{align}
	\widehat{W}=&\frac{1}{d}\ket{000}\bra{000}-[\ket{001}(\bra{001}\nonumber\\+&\bra{100})+\ket{010}(\bra{001}+\bra{100})\nonumber\\+&\ket{100}(\bra{001}+\bra{010})]\nonumber\\
	+&d(\ket{011}+\ket{101}+\ket{110})(\bra{011}\nonumber\\+&\bra{101}+\bra{110}).\nonumber
	\end{align} 
Ref. \cite{chen2012estimating} has shown that the state $\r(p)$ is fully separable when $p\le \frac{\sqrt{3}}{8+\sqrt{3}},$  they also showed when $p\in (\frac{\sqrt{3}}{8+\sqrt{3}},0.2095],$ the state $\r(p)$ is biseparable but not fully separable. Here if Conjecture \ref{cj:nqubitMES=w} was true, then $MC(\rho)> n-2,$ implies that $\r(p)$ is a genuinely entangled state. In particular, this is true when $n=3$ by Theorem \ref{t5}. Thus proving Conjecture \ref{cj:nqubitMES=w} is meaningful to the investigation of  entanglement dectection.  
		
%\red{to xian: please add the entanglement witnesses realizing the above separability etc, so as to connect the last paragraph introducing witnesses}
	
%\red{to xian: please find out in literatures the entanglement witness (EW), and genuine EW, so as to show the critical points by which when $\r(p)$ turns from genuine entanglement to biseparable state, and from biseparable state to fully separable state, by choosing $\ket{\ph}$ as GHZ, W or other known states in quantum information. I'm not sure whether the EWs have been studied, as it may be not that easy? At the same time you may work out $MC(\r)$ at the points. Then I will continue with this section, by comparing these points and values, to show whether the stronger GE is supported by the stronger $MC(\phi)$. This is a physical model of analyzing $MC(\phi)$. }

\subsection{absolutely maximally entangled state}

Next  we investigate an absolutely maximally entangled states (AMES) in $n$-qubit systems. A  pure multipartite entangled state is called AMES if all reduced density operators obtained by tracing out at least half of the particles are maximally mixed \cite{goyeneche2015absolutely}. 
So the function $MC(\cdot)$ of every AMES is zero, though the converse fails because the $n$-qubit non-AMES state may have separable reduced density operators. An example is the GHZ state. Hence, we have constructed a necessary condition by which a multipartite state is an AMES.  Next we consider the AMES in a $3$-qubit system.
\begin{corollary}\label{r1}
 The sole class of AMES $\ket{\psi}$ in an $3$-qubit system are the states that are LU equivalent to $\ket{GHZ}=\frac{1}{\sqrt{2}}(\ket{000}+\ket{111}).$ 
\end{corollary}
We place the proof of the corollary in the APPENDIX \ref{E}. 

	\subsection{An upper bound of the sum of all bipartite quantum discord }
\indent Here we present an upper bound of the sum of all bipartite quantum discord for a three-qubit pure state. The quantum discord was first presented by Henderson and Vedral \cite{henderson2001classical}, Ollivier and Zurek \cite{ollivier2001quantum} independently. Quantum discord is a measure of nonclassical correlation. It is defined as 
\begin{align*}
\delta_{AB}^{\leftarrow}=&I_{AB}-J_{AB}^{\leftarrow}\nonumber\\
=&I_{AB}-\max_{\{\Pi_x^B\}}(S(\rho_A)-\sum_x p_xS(\rho_A^x)),\nonumber
\end{align*}
where the maximum takes over all the positive-operator-valued-measurements $\{\Pi_x^B\}$ performed on the subsystem B, $p_x=\tr{\Pi_x^B\rho_{AB}\Pi_x^B},$ and $\rho_A^x=\tr_B(\Pi_x^B\rho_{AB}\Pi_x^B)/p_x.$
From its definition, it quantifies at least how much a bipartite state of one system is changed on anverage by the measurement of the other system. In the last decade, there are some results suggesting that quantum discord plays an important role in quantum information and computation tasks \cite{dakic2012quantum,gu2012observing,pirandola2014quantum,liu2017resource,liu2020quantum}. Recently, Guo $et$ $al.$ considered the  complete monogamy relation
for multi-party quantum discord \cite{guo2021monogamy}.\par
\indent Next we recall a conservation law for distributed EoF and quantum discord of a three-qubit pure state  \cite{fanchini2011conservation} , 
\begin{align}\label{ced}
E_{AB}+E_{AC}=&\delta_{AB}^{\leftarrow}+\delta_{AC}^{\leftarrow},
\end{align}
The law depends on the Koashi-Winter (KW) relation $E_{AB}+J_{AC}^{\leftarrow}=S_A$ \cite{koashi2004monogamy}. \\
\indent Here we present an upper bound of the sum of all the bipartite discord for a three-qubit pure state, from equation $(\ref{ced})$, we have
\begin{align}
\hspace{2mm}&\delta_{AB}^{\leftarrow}+\delta_{BC}^{\leftarrow}+\delta_{CA}^{\leftarrow}+\delta_{BA}^{\leftarrow}+\delta_{AC}^{\leftarrow}+\delta_{CB}^{\leftarrow}\nonumber\\
=&E_{AB}+E_{AC}+E_{BC}+E_{BA}+E_{CA}+E_{CB}\nonumber\\
\le& 2\times c_{max}=2c_{max},\nonumber\\
\end{align}
\indent Thus we have quantum discord owns a multi-linear monogamy relation for a three-qubit pure state.
\section{Conclusions}
\label{sec:con}
	\indent Here we have mainly considered the shareablility of the entanglement for a multi-qubit state in terms of the EoF. We have presented that up to the local unitary transformations, the W state is the unique that can reach the upper bound of $MC(\cdot)$ for a three-qubit state, these results may tell us that the entanglement cannot be shared freely for a three-qubit system. We have also picked $10^5$ four-qubit and five-qubit pure states randomly and computed their $MC(\cdot)$, which have verified the Conjecture \ref{cj:nqubitMES=w} numerically. Finally, we also have presented some applications of our results. We think the methods here we used can be generalized to consider the upper bound of the multi-linear monogamy relation in terms of other bipartite entanglement measures such as, R$\acute{e}$nyi entanglement for an $n$-qubit pure state.	 We believe that our results are helpful on the study of monogamy relations for multipartite entanglement systems.\\
\section*{Acknowledgments}
Authors were supported by the NNSF of China (Grant No. 11871089), and the Fundamental Research Funds for the Central Universities (Grant Nos. ZG216S2005).
\bibliographystyle{IEEEtran}
\bibliography{ref}
\section{Appendix}
\subsection{The proof of Theorem \ref{meofw}}\label{A}

\textit{\indent For a three-qubit pure state, the W state reaches the upper bound of multi-linear monogamy relation in terms of EoF.}\\
\\	\begin{proof}
	Here we denote that 
	\begin{align}
	f(x)=&h(\frac{1+\sqrt{1-x}}{2})\nonumber\\
	x=&C^2_{AB}\nonumber\\
	y=&C^2_{AC}+C^2_{AB},\nonumber\\
	c=&C^2_{AC}+C^2_{AB}+C^2_{BC},\nonumber\\
	g(x,y)=&f(x)+f(y-x)+f(c-y),\label{g}\\
	\frac{\partial g}{\partial x}=& f^{'}(x)-f^{'}(y-x)=0,\nonumber\\\frac{\partial g}{\partial y}=&f^{'}(y-x)-f^{'}(c-y)=0.\label{pf}
	\end{align}
	As $f^{''}(x)<0,$ $f^{'}(x)$ is monotonously decreasing  \cite{bai2014hierarchical},  $-f^{'}(y-x)$ is also monotonously decreasing in terms of $x$, and by the equality (\ref{pf}), we have 
		\begin{align}
		C^2_{AB}=C^2_{AC}=C^2_{BC} \label{fc2} 
	\end{align}
	is the only case when the $(\ref{pf})$ is valid.
	Furthermore, as $f(x)$ is a monotonic function \cite{bai2014hierarchical}, we have when $(\ref{fc2})$ is valid, $E(\rho_{AB})+E(\rho_{AC})+E(\rho_{BC})$ achieves the upper bound for a three-qubit pure state.\\
	\indent From \cite{acin2000generalized}, we have that a three-qubit pure state $\ket{\psi}_{ABC}$ can be written in the generalized Schmidt decomposition:
	\begin{align}
	\ket{\psi}=l_0\ket{000}+l_1e^{i\theta}\ket{100}+l_2\ket{101}+l_3\ket{110}+l_4\ket{111},
	\end{align}
	here $\theta\in [0,\pi),l_i\ge 0,i=0,1,2,3,4,\sum_{i=0}^4 l_i^2=1.$ From simple computation, we have
	\begin{align} 
	C_{AB}^2=&4l_0^2l_2^2,C_{AC}^2=4l_0^2l_3^2,\\
	C_{BC}^2=&4l_2^2l_3^2+4l_1^2l_4^2-8l_1l_2l_3l_4cos\theta. \label{cc3}
	\end{align}
	As $f(x)$ is monotone \cite{bai2014hierarchical}, then we only need to obtain the maximum of $4l_0^2l_2^2$ by using the Lagrange multiplier,
	\begin{align}
	&m(l_0,l_1,l_2,l_4,\lambda,\mu)=4l_0^2l_2^2+\lambda(l_0^2+l_1^2+2l_2^2+l_4^2-1),\nonumber\\\hspace{4mm}&+\mu(l_0^2l_2^2-l_2^4-l_1^2l_4^2+2l_1l_2^2l_4\cos\theta),\label{me}\\
	&\frac{\partial m}{\partial l_0}=8l_0l_2^2+2\lambda l_0+2\mu l_0l_2^2,\label{f1}\\
	&\frac{\partial m}{\partial l_1}=2\lambda l_1+2\mu l_2^2l_4\cos\theta-2\mu l_1l_4^2,\label{f2}\\
	&\frac{\partial m}{\partial l_2}=8l_0^2l_2+4\lambda l_2+2\mu l_0^2l_2-4\mu l_2^3+4\mu l_1l_2l_4\cos\theta,\label{f3}\\
	&\frac{\partial m}{\partial l_4}=2\lambda l_4-2\mu l_1^2l_4+2\mu l_1l_2^2\cos\theta,\label{f4}\\
	&\frac{\partial m}{\partial \theta}=-2\mu l_1l_2^2l_4\sin\theta,\label{f5}\\
	&\frac{\partial m}{\partial \lambda}=l_0^2+l_1^2+2l_2^2+l_4^2-1,\label{f6}\\
	&\frac{\partial m}{\partial \mu}=l_0^2l_2^2-l_2^4-l_1^2l_4^2+2l_1l_2^2l_4\cos\theta , \label{f7}
	\end{align}
	when the fomulas $(\ref{f1})-(\ref{f7})$ equal to 0, we have $l_0=l_2=l_3=\frac{1}{\sqrt{3}},l_1=l_4=0$ is the only case when $C^2(\rho_{AB})$ attains the maximum, that is, 
	\begin{align} \hspace{2mm}&\max_{\ket{\psi}_{ABC}}(E(\rho_{AB})+E(\rho_{AC})+E(\rho_{BC}))\nonumber\\:=&c_{max}=3h(\frac{1}{2}+\frac{\sqrt{5}}{6}).\label{mp}
	\end{align}
	\par When computing $(\ref{f1})-(\ref{f7})$ equal to 0, according to $(\ref{f5})$, we have that at least one of the equalities in the set $\{\mu=0,l_1=0,l_2=0,l_4=0, \sin\theta=0\}$ is valid.  Then by using the method of exclusion, we could get the result. By the way, in the method of exclusion, we mainly use that when $\ket{\psi}$ is separable, the function $ME(\cdot)$ cannot get the maximum. \\
	\indent When we take the operation $\sigma_x$ on the first system, we get the $W$ state $\frac{1}{\sqrt{3}}(\ket{001}+\ket{010}+\ket{100}).$\\
\end{proof}
\subsection{The proof of Lemma \ref{mulc}}\label{B}
\textit{\indent	Up to the local unitary transformations, the W state is the unique state that can reach the upper bound in terms of the function $MC(\psi)=C_{AB}+C_{BC}+C_{AC}$ for a three-qubit pure state.}\\
\begin{proof}
	We'll compute the maximum of the six classes of a three-qubit state respectively according to \cite{allen2017entanglement}.\\
	\indent $Case$ $i$: When $\ket{\psi}_{ABC}$ is  $A-B-C$, $MC(\psi)=0$.\\
	\indent $Case$ $ii$: When $\ket{\psi}_{ABC}$ is biseparable, if $\ket{\psi}_{ABC}=\ket{\phi_1}_A\otimes \ket{\phi_2}_{BC}$, $C_{AB}=C_{AC}=0$, $MC(\psi)=C_{BC}\le 1$, the other cases are similar. \\
	\indent $Case$ $iii$: When $\ket{\psi}_{ABC}$ belongs to the $W$ class, according to the fomula $(\ref{LUW}),$  $MC(\psi)=2r_2r_3+2r_1r_3+2r_1r_2$. Trivially, when $r_1=r_2=r_3$, that is, $\ket{\psi}_{ABC}=\ket{W}=\frac{1}{\sqrt{3}}(\ket{001}+\ket{010}+\ket{110})$, $MC(\psi)$ gets the maximum.\\
	\indent $Case$ $iv$: When $\ket{\psi}_{ABC}$ belongs to the  $GHZ$ class,  according to \cite{allen2017entanglement},  assume $\ket{\psi}=M_1\otimes M_2\otimes M_3\ket{GHZ}$, $M_i=\big(\vec{u_i},\vec{v_i}\big)$,  $\vec{u_i}=(u_i\cos\theta_i,u_i\sin\theta_i)^{T}$,  $\vec{v_i}=(v_i\cos(\phi_i+\theta_i),v_i\sin(\phi_i+\theta_i))^{T}$,
	\begin{align}\label{s1}
	MC(\psi)=\frac{|c_1s_2s_3|+|c_2s_1s_3|+|c_3s_1s_2|}{r+c_1c_2c_3},
	\end{align}
	here we denote $c_i=\cos\phi_i,s_i=\sin\phi_i,i=1,2,3$, $2r=\frac{u_1u_2u_3}{v_1v_2v_3}+\frac{v_1v_2v_3}{u_1u_2u_3}$,  $r\ge 1$.  In order to let $MC$  be the maximum,  assume $r=1$.   When $c_i<0$,  $MC$ will gets maximum.  then let $\phi_i\in [0,\frac{\pi}{2}],MC(\psi)=\frac{c_1s_2s_3+c_2s_1s_3+c_3s_1s_2}{1-c_1c_2c_3}$. Next we will prove $MC(\psi)\le 2$,  first we define $l(c_1,c_2,c_3)$ as follows,
	\begin{align*}
	&l(c_1,c_2,c_3)\\=&c_1\cos(\phi_2-\phi_3)+c_2\cos(\phi_1-\phi_3)+c_3\cos(\phi_1-\phi_2)\\
	-&c_1c_2c_3,
	\end{align*}
	assume $c_1>c_2\ge c_3$, then we obtain
	\begin{align*}
	l(c_1,c_2,c_3)\le l(c_1,c_1,c_1).
	\end{align*}
	When $c_1=c_2=c_3$,  the function $l(c_1,c_1,c_1)$ is a monotonic function of $c_1$. When $c_1\rightarrow 1$, the function $l(c_1,c_1,c_1)$ gets the maximum, that is, $MC(\psi)\rightarrow 2$.  Then we prove that if $\ket{\psi}_{ABC}$ is a GHZ class state, then $MC(\ket{\psi}_{ABC})\le 2.$ 
	\indent However, from the above analysis, when $MC(\ket{\psi}_{ABC})=2,$ we have $c_1=c_2=c_3\rightarrow 1,$ that is, the matrix $M_i$ is sigular, this is impossible.

\end{proof}
\subsection{The proof of Theorem \ref{mulcm}}\label{C}
\indent 	\textit{Up to the local unitary transformations, the W state is the unique state that can reach the upper bound in terms of the function $MC(\cdot)$ for a three-qubit state.}\\
\begin{proof}
	Combining with Lemma \ref{mulc}, we only need to present that the mixed states cannot reach the upper bound of the multi-linear monogamy relations in terms of concurrence. \\
	\indent Due to Lemma \ref{mulc}, for a three-qubit pure state $\ket{\psi},$  $MC(\ket{\psi})$ gets the maximum, only when $\ket{\psi}$ is LU equivalent to $\ket{W}$. Assume $\rho$ is a three-qubit mixed system, $\{(p_i,\ket{\phi_i})|i=1,2,\cdots,k\}$ is an optimal decompostion of $\rho$, we can always assume $k=2.$ For the cases when $k>2,$ we can prove similarly. As $\ket{\phi_i}$ is LU equivalent to $\ket{W}$, we can always assume $\{(p_1,\ket{W}),(p_2,U_1\otimes U_2\otimes U_3\ket{W})\}$ is a decompostion of $\rho,$ then
	\begin{align}
	MC(\rho)=&C(\rho_{AB})+C(\rho_{AC})+C(\rho_{BC})\nonumber\\
	=&C(\sigma_{1})+C(\sigma_2)+C(\sigma_3)\nonumber\\
	\end{align}
	here we assume 
	\begin{align*}
	\sigma_{1}=&\frac{p_1}{3}(\ket{00}\bra{00}+2\ket{\phi^{+}}\bra{\phi^{+}})+\frac{p_2}{3}\tau_1,\\
	\sigma_{2}=&\frac{p_1}{3}[\ket{00}\bra{00}+2\ket{\phi^{+}}\bra{\phi^{+}}]+\frac{p_2}{3}\tau_2,\\ 
	\sigma_{3}=&\frac{p_1}{3}[\ket{00}\bra{00}+2\ket{\phi^{+}}\bra{\phi^{+}}]+\frac{p_2}{3}\tau_3,\\
	\ket{\phi^{+}}=&\frac{1}{\sqrt{2}}(\ket{01}+\ket{10}),
	\end{align*}
	here we denote $\tau_1=[(U_1\otimes U_2)\ket{00}\bra{00}(U_1\otimes U_2)^{\dagger}+2(U_1\otimes U_2)\ket{\phi^{+}}\bra{\phi^{+}}(U_1\otimes U_2)^{\dagger}],$ $\tau_2=[(U_1\otimes U_3)\ket{00}\bra{00}(U_1\otimes U_3)^{\dagger}+2(U_1\otimes U_3)\ket{\phi^{+}}\bra{\phi^{+}}(U_1\otimes U_3)^{\dagger}],$ $\tau_3=[(U_2\otimes U_3)\ket{00}\bra{00}(U_2\otimes U_3)^{\dagger}+2(U_2\otimes U_3)\ket{\phi^{+}}\bra{\phi^{+}}(U_2\otimes U_3)^{\dagger}],$
	then we have 
	\begin{align}
	C(\sigma_{1})\le \frac{2}{3}C(p_1\ket{\phi^{+}}\bra{\phi^{+}}+p_2(U_1\otimes U_2)\ket{\phi^{+}}\bra{\phi^{+}}(U_1\otimes U_2)^{\dagger}),
	\end{align}
By the Lemma \ref{SD}, \ref{rank1} and \ref{rank2}, we have $(U_1\otimes U_2)\ket{\phi^{+}}=e^{ix}\ket{\phi^{+}}$ is a sufficient and necessary condition of $C(p_1\ket{\phi^{+}}\bra{\phi^{+}}+p_2(U_1\otimes U_2)\ket{\phi^{+}}\bra{\phi^{+}}(U_1\otimes U_2)^{\dagger})=1,$ here $e^{ix}$ is a global phase factor.  Then we can get the similar result for $\sigma_{2}$ and $\sigma_3.$ As $Sr(\ket{01})=Sr(\ket{10})=1$, here we denote that $Sr(\cdot)$ is the schmidt rank, then we have $U_i=\left(\begin{array}{cc}
	e^{i\theta_i}&0\\
	0&1
	\end{array}\right),i=1,2,3$, that is, $\rho=\ket{W}\bra{W}.$ Then we finish the proof.\\
\end{proof}

\indent Here we will prove that $(U_1\otimes U_2)\ket{\phi^{+}}=e^{ix}\ket{\phi^{+}}$ is a sufficient and necessary condition of $C(p_1\ket{\phi^{+}}\bra{\phi^{+}}+p_2(U_1\otimes U_2)\ket{\phi^{+}}\bra{\phi^{+}}(U_1\otimes U_2)^{\dagger})=1.$ As $\Longrightarrow$ is trivial, \\
\indent $\Longleftarrow:$ First we will present a lemma.
\begin{lemma}\label{SD}
	Assume $A, B\in Pos(\mathbb{H})$,  here we denote that $Pos(\mathbb{H})$ is a linear space consisting of all the semidefinite positive operators of a bounded Hilbert space $\mathbb{H}$. Here we denote that $Eig(A)$ and $Eig(B)$ are two sets consisting of all the eigenvalues of the matrix $A$ and $B$ respectively. If the biggest elements in the set $Eig(A)$ and $Eig(B)$ are 1 or less, then the biggest element in the set $Eig(AB)$ is 1 or less.
\end{lemma}
\begin{proof}
	Assume that the eigenvalues of $A$ are $\lambda_i$ with its eigenvector $\ket{\alpha_i}$, the eigenvalues of $B$ are $\mu_j$ with its eigenvector $\ket{\beta_j}$, and the eigenvalues of $AB$ are $\chi_k$ with its eigenvectors $\ket{\gamma_k}$.  Here we always assume that the range of $A$ and $B$ are nonsingular, then we have
	\begin{align}
	AB\ket{\gamma_k}=&AB\sum_{j}x_{jk}\ket{\beta_j}\nonumber\\
	=&A\sum_{j}\mu_j x_{jk}\ket{\beta_j}\nonumber\\
	=&\sum_{ij}\lambda_i\mu_j x_{jk}y_{ij}\ket{\alpha_i},\label{bet}\\
	AB\ket{\gamma_k}=&\chi_k\ket{\gamma_k}\nonumber\\
	=&\chi_k\sum_{ij}x_{jk}y_{ij}\ket{\alpha_i} ,\label{gamma}
	\end{align}
	in the formula $(\ref{bet}),$ we denote that $\ket{\gamma_k}=\sum_j x_{jk}\ket{\beta_j}$ and $\ket{\beta_j}=\sum_i y_{ij}\ket{\alpha_i}.$  From the equality $(\ref{bet})$ and $(\ref{gamma}),$ we have
	\begin{align}
\chi_k=\frac{\sum_{ij}\lambda_i\mu_j x_{jk}y_{ij}}{\sum_{ij}x_{jk}y_{ij}}\le 1. \label{chi}
\end{align}
here the $(\ref{chi})$ is due to  $\lambda_i\le 1$ and $\mu_j\le 1$. Then we finish the proof.
\end{proof}\par
\indent As $\rho$ is semidefinite positive , then $(\sigma_y\otimes \sigma_y)\overline{\rho}(\sigma_y\otimes\sigma_y)$ is semidifinite positive, then due to the Lemma \ref{SD}, we have that all the eigenvalues of $\rho(\sigma_y\otimes \sigma_y)\overline{\rho}(\sigma_y\otimes\sigma_y)$ are 1 or less. Then due to $(\ref{C2}),$ we have that only when $Rank(\rho(\sigma_y\otimes \sigma_y)\overline{\rho}(\sigma_y\otimes\sigma_y))=1$ and $\lambda_1$ in $(\ref{C2})$ equals to 1 can $C(\rho)=1.$ As $\sigma_y\otimes \sigma_y$ is nonsigular, we only need $Rank(\rho)=1.$ \\
\begin{lemma} \label{rank1}
	$Rank[p_1\ket{\phi^{+}}\bra{\phi^{+}}+p_2(U_1\otimes U_2)\ket{\phi^{+}}\bra{\phi^{+}}(U_1\otimes U_2)^{\dagger}]=1$ if and only if $(U_1\otimes U_2)\ket{\phi^{+}}=e^{ix}\ket{\phi^{+}},$ here $e^{ix}$ is a global phase factor. 
	\end{lemma}
\begin{proof}Here we denote $\sigma=p_1\ket{\phi^{+}}\bra{\phi^{+}}+p_2(U_1\otimes U_2)\ket{\phi^{+}}\bra{\phi^{+}}(U_1\otimes U_2)^{\dagger}.$  If $(U_1\otimes U_2)\ket{\phi^{+}}\ne e^{ix}\ket{\phi^{+}},$ then $(U_1\otimes U_2)\ket{\phi^{+}}$ and $\ket{\phi^{+}}$ are linear independent,  $dim(span\{\ket{\phi^{+}},(U_1\otimes U_2)\ket{\phi^{+}}\})=2,$ $dim(span\{\ket{\phi^{+}},(U_1\otimes U_2)\ket{\phi^{+}}\}^{\perp}=n-2.$ As $\forall \ket{\alpha}\in dim(span\{\ket{\phi^{+}},(U_1\otimes U_2)\ket{\phi^{+}}\}^{\perp},$ $\sigma\ket{\alpha}=0,$ that is, $Rank(\sigma)\le 2.$  As we cannot find a nontrivial  vector $\ket{\beta}$ in the subspace ${span\{\ket{\phi^{+}},(U_1\otimes U_2)\ket{\phi^{+}}\}}$ such that $\sigma\ket{\beta}=0,$ then we finish the proof.
\end{proof}

\begin{lemma}\label{rank2}
	Assume $\theta=p_1\ket{\phi^{+}}\bra{\phi^{+}}+p_2U_1\otimes U_2)\ket{\phi^{+}}\bra{\phi^{+}}(U_1\otimes U_2)^{\dagger}$ with $Rank(\theta)=2,$  then $Rank[(\sigma_y\otimes \sigma_y)\tilde{\theta}(\sigma_y\otimes \sigma_y)\theta]=2.$
	\end{lemma}
\begin{proof}
	As $\sigma_y\otimes\sigma_y$ is invertible, we only need to prove $Rank[\tilde{\theta}(\sigma_y\otimes \sigma_y)\theta]=2.$ As $Rank[\tilde{\theta}(\sigma_y\otimes \sigma_y)\theta]=Rank[\theta\tilde{\theta}(\sigma_y\otimes \sigma_y)],$ then if we can prove $Rank[\theta\tilde{\theta}]=2,$ we finish the proof. As  $\theta=\theta^{\dagger},$ then $Rank[\theta\tilde{\theta}]=Rank[\theta\theta^{T}]=Rank(\theta)=2.$ 
\end{proof}
\subsection{The proof of Theorem \ref{6}}\label{D}
\indent \textit{Assume $\ket{\psi}$ is a three-qubit pure state, then $MC(\ket{\psi})=0$ if and only if $\ket{\psi}$ can be represented as $\ket{\psi}=r_0\ket{000}+r_1\ket{111}$ up to local unitary operations when $0\le r_0,r_1\le 1.$}\\
\begin{proof}
	First we recall $MC(\ket{\psi}_{ABC})=C^2_{AB}+C_{AC}^2+C_{BC}^2.$\\
	{$\Leftarrow:$ When $\ket{\psi}=r_0\ket{000}+r_1\ket{111}$, then $\rho_{AB}=\rho_{AC}=\rho_{BC}=r^2_0\ket{00}\bra{00}+r_1^2\ket{11}\bra{11}$, $C(\rho_{AB})=C(\rho_{AC})=C(\rho_{BC})=0,$ that is, $MC(\ket{\psi})=0.$\\
	$\Rightarrow:$ In the proof of theroem \ref{meofw}, we present that for a three-qubit pure state $\ket{\psi}_{ABC}=l_0\ket{000}+l_1e^{i\theta}\ket{100}+l_2\ket{101}+l_3\ket{110}+l_4\ket{111}$, here $l_i\ge 0,$ $i=0,1,2,3,4,$ $\theta\in[0,\pi),$
	\begin{align}
	C_{AB}^2=&4l_0^2l_2^2,\hspace{2mm} C_{AC}^2=4l_0^2l_3^2,\nonumber\\
	C_{BC}^2=&4l_2^2l_3^2+4l_1^2l_4^2-8l_1l_2l_3l_4\cos\theta. \label{ccc}
	\end{align}
	When $MC(\ket{\psi})=0,$ we have that 
	\begin{align}
	l_0l_2=l_0l_3=0,\nonumber\\
	l_2^2l_3^2+l_1^2l_4^2=2l_1l_2l_3l_4\cos\theta.
	\end{align}
	When $l_0\ne 0,$ we have $l_2=l_3=0,$ that is, $l_1l_4=0,$  then $\ket{\psi}=l_0\ket{000}+l_4\ket{111},$ or $\ket{\psi}=(l_0\ket{0}+l_1e^{i\theta}\ket{1})\ket{00}.$ When $\ket{\psi}$ is the second state, let $U_1$ be a unitary on the first system such that $U_1(\ket{l_0\ket{0}+l_1e^{i\theta}\ket{1}})=\ket{0},$ then the second state is LU equivalent to the state $\ket{000}.$ Below we denote $U_i$ be a unitary on the $i$-th system, $i=1,2,3$.\\
\par
 When $l_0=0,$ then from the third formula in (\ref{ccc}), 
\begin{align}
(l_2l_3\cos\theta-l_1l_4)^2+l_2^2l_3^2(\sin\theta)^2=0,
\end{align}
that is, $l_2l_3\cos\theta=l_1l_4,$ $l_2l_3\sin\theta=0.$ When $\sin\theta=0$ and $l_i\ge 0,$ $i=1,2,3,4,$ $\cos\theta=1.$ That is, $(l_1,l_2)=a(l_3,l_4),$ then $\ket{\psi}=\ket{1}(\ket{0}+a\ket{1})(l_1\ket{0}+l_2\ket{1}).
$ When $U_2(\ket{0}+a\ket{1})=\sqrt{1+a^2}\ket{0},$ $U_3(l_1\ket{0}+l_2\ket{1})=\sqrt{l_1^2+l_2^2}\ket{0},$ then we obtain that the above state is LU equivalent to $\ket{000}.$ \hspace{3mm} When $l_2=0,$ then $l_1l_4=0.$ If $l_1=0,$ then $\ket{\psi}$ can be represented as $\ket{\phi_1}=l_3\ket{110}+l_4\ket{111}.$ Let $U_3(l_3\ket{0}+l_4\ket{1})=\ket{1}$, then it is LU equivalent to the state $\ket{111},$ If $l_4=0,$ then $\ket{\psi}$ can be represented as $\ket{\phi_2}=e^{i\theta}l_1\ket{100}+l_3\ket{110}$. Let $U_2(e^{i\theta}l_1\ket{0}+l_3\ket{1})=\ket{1}$ and $U_3=\sigma_X$, then it is LU equivalent to the state $\ket{111}.$ \indent The case when $l_3=0$ is similar to the case when $l_2=0.$  Then we finish the proof. }
	\end{proof}
\subsection{The proof of Corollary \ref{r1}}\label{E}
\textit{The sole class of AMES $\ket{\psi}$ in an $3$-qubit system are the states that are LU equivalent to $\ket{GHZ}=\frac{1}{\sqrt{2}}(\ket{000}+\ket{111}).$ }\par
\begin{proof}
	First we provide two methods to prove that $\rho_{AB}$, $\rho_{AC}$ and $\rho_{BC}$ are separable.
	Assume $\ket{\psi}_{ABC}$ is an AMES state in an $3$-qubit state, then $$\rho_A=\rho_B=\rho_C=\frac{I}{2}.$$ Next as $\rho_A=\frac{I}{2},$ then any purification state $\ket{\phi}_{AB^{'}}$ of $\rho_A$ can be written as $$\ket{\phi}_{ABC}=(I_A\otimes U_{BC})\frac{(\ket{00}+\ket{11})\ket{0}}{\sqrt{2}},$$
	here $U_{BC}$ is a unitary operator, then we have 
	\begin{align}
	\rho_{BC}=\ket{\phi_1}\bra{\phi_1}+\ket{\phi_2}\bra{\phi_2},\nonumber\\
	\ket{\phi_1}=U_{BC}\ket{00},\hspace{3mm} \ket{\phi_2}=U_{BC}\ket{10}
	\end{align} 
	that is, $r(\rho_{BC})=2,$ then due to Theorem 1 in \cite{kraus2000separability}, we have $\rho_{BC}$ is separable. Similarly, we have $\rho_{AB}$ and $\rho_{AC}$ is separable.\par
	Here we provide the other method to prove that $\rho_{AB},$ $\rho_{AC}$ and $\rho_{BC}$ are separable. As $\rho_A=\rho_B=\frac{I}{2},$ then from \cite{kraus2010local}, we have
	\begin{align}
	\rho_{AB}=&\lambda_1\Psi^{+}+\lambda_2\Psi^{-}+\lambda_3\Phi^{+}+\lambda_4\Phi^{-},\\
	\Psi^{+}=&\frac{1}{2}(\ket{00}+\ket{11})(\bra{11}+\bra{00}),\nonumber\\
	\Psi^{-}=&\frac{1}{2}(\ket{00}-\ket{11})(\ket{00}-\bra{11}),\nonumber\\
	\Phi^{+}=&\frac{1}{2}(\ket{01}+\bra{10})(\bra{01}+\bra{10}),\nonumber\\
	\Phi^{-}=&\frac{1}{2}(\ket{01}-\ket{10})(\bra{01}-\bra{10})\nonumber.
	\end{align}
	As $\rho_C$ and $\rho_{AB}$ are with the same spectrum, then we have only two of $\lambda_i,$ $i=1,2,3,4$ are $\frac{1}{2}$. Then $\rho_{AB}$ is separable. Similarly, we have $\rho_{AC}$ and $\rho_{BC}$ are separable.\par
	As all of $\rho_{AB},$ $\rho_{AC}$ and $\rho_{BC}$ are separable, then we have
	\begin{align}
	MC(\ket{\psi})=0\nonumber,
	\end{align} 
	then from Theorem \ref{6}, we have $\ket{\psi}=r_0\ket{000}+r_1\ket{111}$ up to unitary operations when $r_0,r_1\in [0,1].$ As $\rho_A=\rho_B=\rho_C=\frac{1}{2},$ then $\ket{\psi}=\frac{1}{2}(\ket{000}+\ket{111})$ up to local unitary operations. 
\end{proof}
 	\end{document}